\documentclass {article}
\usepackage{fullpage}

\usepackage{graphics}

\usepackage{amsmath}
\usepackage{amssymb}
\usepackage{amsfonts}
\usepackage{graphicx}

\usepackage{authblk}

\usepackage{array}
\usepackage{color}

\usepackage{algorithm}

\usepackage[english]{babel}

\usepackage{multirow}
\usepackage{rotating}

\usepackage{epstopdf}

\usepackage{authblk}

\begin{document}

\newtheorem{theorem}{Theorem}[section]
\newtheorem{lemma}{Lemma}[section]
\newtheorem{corollary}{Corollary}[section]
\newtheorem{claim}{Claim}[section]
\newtheorem{proposition}{Proposition}[section]
\newtheorem{definition}{Definition}[section]
\newtheorem{fact}{Fact}[section]
\newtheorem{example}{Example}[section]

\newtheorem{remark}{Remark}[section]

\newcommand{\proof}{\em Proof: \em}

\newcommand{\cA}{{\cal A}}
\newcommand{\cC}{{\cal C}}
\newcommand{\cG}{{\cal G}}
\newcommand{\cN}{{\cal N}}
\newcommand{\cU}{{\cal U}}
\newcommand{\cT}{{\cal T}}
\newcommand{\cS}{{\cal S}}
\newcommand{\cL}{{\cal L}}
\newcommand{\cV}{{\cal V}}
\newcommand{\loc}{{\cal LOCAL}}

\newcommand{\ai}{\alpha_i}
\newcommand{\bi}{\beta_i}
\newcommand{\gi}{\gamma_i}
\newcommand{\di}{\delta_i}

\newcommand{\oai}{\overline{\alpha}_i}
\newcommand{\obi}{\overline{\beta}_i}
\newcommand{\ogi}{\overline{\gamma}_i}
\newcommand{\odi}{\overline{\delta}_i}

\newcommand{\algo}[1]{
\medskip
\noindent \textbf{Algorithm {\tt #1}}\\
\nopagebreak}

\newcommand{\procedure}[1]{
\medskip
\noindent \textbf{Procedure {\tt #1}}
\nopagebreak \\}

\newcommand{\procend}{\hfill $\blacksquare$\medskip}

\newcommand{\WLE} {{\tt Weak Leader Elect}}
\newcommand{\SLE} {{\tt Strong Leader Elect}}
\newcommand{\TLE} {{\tt Tree Leader Elect}}
\newcommand{\MT} {{\tt Move Token}}




\title{Knowledge, Level of Symmetry, and Time of Leader Election
\thanks{A preliminary version of this paper appeared in the Proc. of the 20th Annual European Symposium on Algorithms (ESA 2012), LNCS 7501, 479--490.}}

\author[1]{ Emanuele G. Fusco}

\affil[1]{Department of Computer, Control, and Management Engineering ``Antonio Ruberti''
Sapienza, University of Rome.
{\tt fusco@diag.uniroma1.it}}

\author{Andrzej Pelc\thanks{Andrzej Pelc was partially supported by NSERC discovery grant 
and by the Research Chair in Distributed Computing at the
Universit\'e du Qu\'ebec en Outaouais.}}

\affil[1]{D\'epartement d'informatique, Universit\'e du Qu\'ebec en Outaouais, Gatineau,
Qu\'ebec J8X 3X7, Canada.
{\tt pelc@uqo.ca}}

\maketitle

\thispagestyle{empty}

\begin{abstract}
We study the time needed for deterministic leader election in the $\loc$ model, where in every round
a node can exchange any messages with its neighbors and perform any local computations.
The topology of the network is unknown and nodes are unlabeled, but ports at each node have arbitrary fixed labelings which,
together with the topology of the network,
 can create asymmetries to be exploited in leader election.
We consider two versions  of the leader election problem: strong LE in which exactly one leader has to be elected, if this is possible, while
all nodes must terminate  declaring that leader election is impossible otherwise,
 and weak LE, which differs from strong LE in that no requirement on the behavior of nodes is imposed,
if leader election is impossible. 
We show that the time of leader election depends on three parameters of the network: its diameter $D$,  its size $n$, and its {\em level of symmetry} $\lambda$,
which, when leader election is feasible, is the smallest depth at which some node has a unique view of the network. 
It also depends on the  knowledge by the nodes, or lack of it,
of parameters $D$ and $n$. 

\vspace*{0.5cm}
\noindent 
{\bf keywords:} leader election,  anonymous network, level of symmetry

\end{abstract}


\section{Introduction}

\subsection{The model and the problem}

Leader election is one of the fundamental problems in distributed computing, first stated in \cite{LL}. Every node
of a network has a boolean variable initialized to 0 and, after the election, exactly one
node, called the {\em leader}, should change this value to 1. All other nodes should know which node is the leader.
If nodes of the network have distinct labels, then leader election is always possible
(e.g., the node with the largest label can become a leader).
However, nodes may refrain from revealing their identities, e.g., for security or privacy reasons.
Hence it is desirable to have leader election algorithms that do not rely on node identities but exploit
asymmetries of the network due to its topology and to port labelings. With unlabeled nodes, leader election is impossible in symmetric networks{~\cite{Ang80}}.

A network is modeled as an undirected connected graph.
We assume that nodes are unlabeled, but ports at each node have arbitrary fixed labelings $0,\dots,$ $d-1$, where $d$ is the degree of the node. Throughout the paper, we will use the term ``graph'' to mean a graph with the above properties.  
We do not assume any coherence between port labelings at various nodes. Nodes can read the port numbers.  When sending a message through port
$p$, a node adds this information to the message, and when receiving a message through port $q$, a node is aware that this message came through this port. 
The topology of the network is unknown to the nodes, but depending on the 
specification of the problem, nodes may know some numerical parameters of the network, such as the number $n$ of nodes (size), and/or the diameter $D$.
We consider two versions  of the leader election problem (LE): \emph{strong LE} and \emph{weak LE}.
In strong LE one leader has to be elected whenever this is possible, while all nodes must terminate  declaring that leader election is impossible otherwise.
Weak LE differs from strong LE in that no requirement on the behavior of nodes is imposed, if leader election is impossible. In both cases, upon election of the leader, every non-leader
is required to know a path (coded as a sequence of ports) from it to the leader.

In this paper we investigate the time of leader election in the extensively studied $\loc$ model \cite{Pe}. In this model, communication proceeds in synchronous \\
rounds and all nodes start simultaneously. In each round each node
can exchange arbitrary messages with all its neighbors and perform arbitrary local computations. 
The time of completing a task is the number of rounds it takes.
{In particular, the time of weak LE is the number of rounds required by the last node to elect a leader if leader election is possible, and the time of strong LE is the number of rounds required by the last node to elect a leader if this is possible, and to terminate by declaring that this is impossible, otherwise.
}

It should be observed that the synchronous process of the $\loc$ model can be simulated in an asynchronous network, 
by defining for each node separately its asynchronous round $i$
in which it performs local computations, then sends messages stamped $i$ to all neighbors, and  waits until getting messages stamped $i$ from all neighbors.
(To make this work, every node is required to send a message with all consecutive stamps, until termination, possibly empty messages for some stamps.)
All our results concerning time can be translated for asynchronous networks by replacing ``time of completing a task''  by ``the maximum number of asynchronous rounds  to complete it, taken over all nodes''.

If nodes have distinct labels, then time $D$ in the $\loc$ model is enough to solve any problem solvable on a given network, as after this time all nodes have an exact
copy of the network. By contrast, in our scenario of unlabeled nodes, time $D$ is often not enough, for example to elect a leader, even if this task is feasible.
This is due to the fact that  after time $t$ each node learns only all paths of length $t$ originating at it and coded as sequences of port numbers. This is far less information
than having a picture of the radius $t$ neighborhood. A node $v$ may not know if two paths originating at it have the same other endpoint or not. It turns out
that these ambiguities may force time much larger than $D$ to accomplish leader election.

We show that the time of leader election depends on three parameters of the network: its diameter $D$,  its size $n$, and its 
{\em level of symmetry} $\lambda$.
The latter parameter is defined for any network (see Section 2 for the formal definition) and,
if leader election is feasible, this is the smallest depth at which some node has a unique view of the network. The view at depth $t$ from a node
(formally defined in Section 2) is the tree of all paths of length $t$ originating at this node, coded as sequences of port numbers on these paths.
It is the maximal information a node can gain after $t$ rounds of communication in the $\loc$ model. 

It turns out that the time of leader election also crucially depends on the {\em knowledge} of parameters $n$ and/or $D$ by the nodes. On the other hand, it does not depend on {\em knowing} $\lambda$, although it often depends on this parameter as well.

\subsection{Our results}
Optimal time of weak LE is shown to be $\Theta(D+\lambda)$, if either $D$ or $n$ is known to the nodes. 
More precisely, we give two algorithms, one working for the class of networks of given diameter $D$ and the other for the class of networks of given
size $n$, that elect a leader in time $O(D+\lambda)$ on networks with diameter $D$ and level of symmetry $\lambda$, whenever election is possible. 
Moreover, we prove that this complexity cannot be improved.  We show, for any 
values $D$ and $\lambda$, a network of diameter $D$ and level of symmetry $\lambda$ on which leader election is possible but takes time at least $D+\lambda$,
even when $D$, $n$ and $\lambda$ are known.
If neither $D$ nor $n$ is known, then even weak LE is impossible \cite{YK3}.

\begin{table*}
\begin{center}
   \begin{tabular}{cccccc}
  \cline{3-6} 
 \multicolumn{2}{c}{} &\multicolumn{4}{|c|}{ {\sc KNOWLEDGE} } \cr \cline{3-6} 
\multicolumn{2}{c}{} & \multicolumn{1}{|c}{\sc none} & \multicolumn{1}{|c}{\sc diameter} & \multicolumn{1}{|c}{\sc size} & \multicolumn{1}{|c|}{\sc diameter \& size} \\ \hline
\multicolumn{1}{|c}{\multirow{4}{*}{\begin{sideways}{\sc TASK}\end{sideways}}}
&\multicolumn{1}{|c}{\multirow{2}{*}{{WLE}}} &

	\multicolumn{1}{|c}{{\sc impossible}} &  \multicolumn{3}{|c|}{$\Theta(D+\lambda)$}\\
\multicolumn{1}{|c}{} & \multicolumn{1}{|c}{} &
	\multicolumn{1}{|c}{{\sc \cite{YK3}}} &  \multicolumn{3}{|c|}{(fast)}
	\cr  \cline{2-6}

\multicolumn{1}{|c}{}&\multicolumn{1}{|c}{\multirow{2}{*}{{SLE}}} &
	\multicolumn{2}{|c}{\multirow{2}{*}{{\sc impossible}}} &  \multicolumn{1}{|c}{$\Theta(n)$}
	& \multicolumn{1}{|c|}{$\Theta(D+\lambda)$} \\
	
\multicolumn{1}{|c}{} & \multicolumn{1}{|c}{} &	
	\multicolumn{2}{|c}{} &  \multicolumn{1}{|c}{(slow)}
	& \multicolumn{1}{|c|}{(fast)} 
	\\ \hline	
\\
\end{tabular}
\caption{\label{summary}Summary of the results}
\end{center}
\end{table*}

For strong LE, we show that knowing only $D$ is insufficient to perform it.  Then we prove that, if only $n$ is known, then
optimal time is $\Theta(n)$. We give an algorithm working for the class of networks of given
size $n$, which performs strong LE in time $O(n)$ and we show, for arbitrarily large $n$, a network $G_n$ of size $n$, diameter $O(\log n)$ and
level of symmetry 0, such that any algorithm
performing correctly strong LE on all networks of size $n$ must take time $\Omega(n)$ on $G_n$. Finally, 
if both $n$ and $D$ are known, then optimal time is $\Theta(D+\lambda)$. Here we give an algorithm, working for the class of networks of
given size $n$ and given diameter $D$ which performs strong LE in time $O(D+\lambda)$ on networks with level of symmetry $\lambda$.
In this case the matching lower bound carries over from our result for weak LE. Table~\ref{summary} gives a summary of our results. The main difficulty of this study
is to prove lower bounds on the time of leader election, showing that the complexity of the proposed algorithms cannot be improved for any of the considered scenarios. 

The comparison of our results for various scenarios shows two exponential gaps. The first is between the time of strong and weak LE.
When only the size $n$ is known, strong LE takes time $\Omega(n)$ on some graphs of logarithmic diameter and level of symmetry 0, while weak LE is accomplished in time $O(\log n)$
for such graphs. The second exponential gap is for the time of strong LE, depending on the knowledge provided.
While knowledge of the diameter alone does not help to accomplish strong LE, when this knowledge is added to the knowledge of the size,
it may exponentially decrease the time of strong LE. Indeed, strong LE with the knowledge of $n$ alone takes time $\Omega(n)$
on some graphs of logarithmic diameter and level of symmetry 0,
while strong LE with the knowledge of $n$ and $D$ is accomplished in time $O(\log n)$, for such graphs.

\subsection{Related work}
Leader election was first studied for rings, under the assumption that all labels are distinct.
A synchronous algorithm, based on comparisons of labels, and using
$O(n \log n)$ messages was given in \cite{HS}. It was proved in \cite{FL} that
this complexity is optimal for comparison-based algorithms. On the other hand, the authors showed
an algorithm using a linear number of messages but requiring very large running time.
An asynchronous algorithm using $O(n \log n)$ messages was given, e.g., in \cite{P} and
the optimality of this message complexity was shown in \cite{B}. Deterministic leader election in radio networks has been studied, e.g., 
in \cite{JKZ,KP,NO} and randomized leader election, e.g., in \cite{Wil}.

Many authors \cite{ASW,AtSn,BV,DKMP,Kr,KKV,Saka,YK,YK3} studied various computing
problems in anonymous networks. In particular, \cite{BSVCGS,YK3} characterize networks in which
leader election can be achieved when nodes are anonymous. In \cite{YK2} the authors study
the problem of leader election in general networks, under the assumption that labels are
not unique. They characterize networks in which this can be done and give an algorithm
which performs election when it is feasible. They assume that the number of nodes of the
network is known to all nodes. In
 \cite{FKKLS}  the authors
study feasibility and message complexity of sorting and leader election in rings with
nonunique labels, while in \cite{DP} the authors provide algorithms for the
generalized leader election problem in rings with arbitrary labels,
unknown (and arbitrary) size of the ring and for both
synchronous and asynchronous communication. In \cite{HKMMJ} the leader election problem is
approached in a model based on mobile agents.
Characterizations of feasible instances for leader election and naming problems have been provided in~\cite{C,CMM,CM}.
Memory needed for leader election in unlabeled networks has been studied in \cite{FP}. To the best of our knowledge, the problem
of time of leader election in arbitrary unlabeled networks has never been studied before.

\section{Preliminaries}\label{sec:preliminaries}

We say that leader election is {\em possible} for a given graph, or that this graph is a {\em solvable graph}, 
if there exists an algorithm which performs LE for this graph.

We consider two versions of the leader election task for a class $\cC$ of graphs :\\
$\bullet$
Weak LE.
Let $G$ be any graph in class $\cC$.  If leader election is possible for the graph $G$, then 
it is accomplished.\\
$\bullet$
Strong LE.  
Let $G$ be any graph in class $\cC$.  If leader election is possible for the graph $G$, then 
it is accomplished.
Otherwise, all nodes eventually declare that the graph is not solvable and stop.


Hence weak LE differs from strong LE in that, in the case of impossibility of leader election, no restriction on the behavior of nodes is imposed:
they can, e.g.,  elect different leaders, or no leader at all, or the algorithm may never stop.




We will use the following notion from \cite{YK3}. Let $G$ be a graph and $v$ a node of $G$. 
{
\begin{definition}[View]
The {\em view} from $v$ is the infinite rooted tree $\cV(v)$ with labeled ports, defined recursively as follows. $\cV(v)$ has the root $x_0$ corresponding to $v$. For every node $v_i$, $i=1,\dots ,k$, adjacent to $v$, 
there is a neighbor $x_i$ in $\cV(v)$ such that the port number at $v$ corresponding to edge $\{v,v_i\} $ is the same as the port number 
at $x_0$ corresponding to edge $\{x_0,x_i\}$,
and the port number at $v_i$ corresponding to edge $\{v,v_i\} $ is the same as the port number at $x_i$ corresponding to edge $\{x_0,x_i\}$. 
Node $x_i$, for $i=1,\dots ,k$, 
is now the root of the view from $v_i$.
\end{definition}}

\begin{figure}
\begin{center}
\scalebox{0.3}{
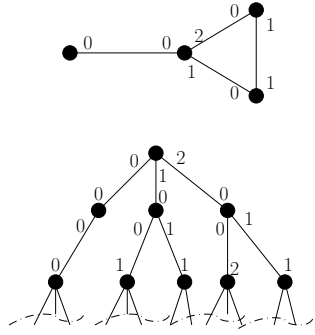}
\caption{\label{fig.view} A graph and a fragment of the corresponding view from its degree 3 node.}
\end{center}
\end{figure}

{See Fig.~\ref{fig.view} for an example of a view.}
   The following proposition
directly follows from \cite{YK3} and expresses the feasibility of leader election in terms of views.



\begin{proposition}\label{possible}
Let $G$ be a graph. The following conditions are equivalent:\\
1. Leader election is possible in $G$;\\ 
2. Views of all nodes are different;\\
3. There exists a node with a unique view.
\end{proposition}

 {In particular, the impossibility of deterministic leader election,
 when there are two nodes with identical views, is implied by proposition 4 from~\cite{YK3},  while corollary~1 from~\cite{YK3} implies that one node has a unique view if and only if all views of the nodes are different.}

By $\cV^t(v)$ we denote the view $\cV(v)$ truncated to depth $t$.
We call it the view of $v$ at depth $t$. In particular, $\cV^0(v)$ consists of the node $v$, together with its degree.
The following proposition was proved in \cite{Norris}.

\begin{proposition}\label{trunc}
For a $n$-node graph,
$\cV(u)=\cV(v)$, if and only if $\cV^{n-1}(u)=\cV^{n-1}(v)$. 
\end{proposition}

For graphs of sublinear diameter a better bound than that of Proposition~\ref{trunc} has been
proved by Hendrickx~\cite{H13} after the publication of the conference version of this paper; for an $n$-node graph of diameter $d$, $\cV(u)=\cV(v)$, if and only if $\cV^{k}(u)=\cV^{k}(v)$, where $k\in O(d+d\log(n/d))$.

Define the following equivalence relations on the set of nodes of a graph.
$u\sim v$ if and only if $\cV(u)=\cV(v)$, and $u\sim_t v$ if and only if $\cV^t(u)=\cV^t(v)$.
Let $\Pi$ be the partition of all nodes into equivalence classes of $\sim$, and $\Pi_t$ the corresponding partition for  $\sim_t$.
It was proved in \cite{YK3} {(corollary~1)} that all equivalence classes in $\Pi$ are of equal size $\sigma$. In view of Proposition \ref{trunc}
this is also the case for $\Pi_t$, where $t\geq n-1$. On the other hand, for smaller $t$, equivalence classes in $\Pi_t$ 
may be of different sizes.  
Every equivalence class in  $\Pi_t$ is a union of some equivalence classes in  $\Pi_{t'}$, for $t<t'$.  
The following result was proved in \cite{Norris} { (lemma~1)}. It says that if the sequence of partitions $\Pi_t$ stops changing at some point, it will never change again. 

\begin{proposition}\label{stop}
If $\Pi_t=\Pi_{t+1}$, then $\Pi_t=\Pi$.
\end{proposition}

For a set $A$, let $|A|$ denote its size. By definition of partitions $\Pi_t$, if $|\Pi_t|=|\Pi_{t+1}|$ then
$\Pi_t=\Pi_{t+1}$. Hence in order to see when partitions stabilize, it is enough to check when their sizes stabilize. 
{This allows us to modify Proposition~\ref{stop} and obtain Proposition~\ref{stopBySize} below, which we will use in some of our algorithms.

\begin{proposition}\label{stopBySize}
If $|\Pi_t|=|\Pi_{t+1}|$, then $\Pi_t=\Pi$.
\end{proposition}}


{\begin{definition}[Level of symmetry]
 For any graph $G$ we define its {\em level of symmetry} $\lambda$ as the smallest integer $t$, for which there exists a node $v$
satisfying the condition $\{u: \cV^t(u)=\cV^t(v)\}=\{u: \cV(u)=\cV(v)\}$. 
\end{definition}}
By Proposition \ref{possible}, for solvable graphs,
the level of symmetry is the smallest $t$ for which there is a node with a unique view at depth $t$. In general, the level of symmetry is the smallest integer $t$ for which some equivalence class in  $\Pi_t$ has size $\sigma$.

{\begin{definition}
Define $\Lambda$ to be the smallest integer $t$ for which $\Pi_t=\Pi_{t+1}$. 
\end{definition}}

We have 

\begin{proposition} \label{Lambda}
$\Lambda \leq D + \lambda$.
\end{proposition}

\begin{proof}
It is enough to show that if $\cV^{D+\lambda}(w)=\cV^{D+\lambda}(w')$, for some nodes $w$ and $w'$, then  $\cV(w)=\cV(w')$.
Let $v$ be a node for which $\{u: \cV^{\lambda}(u)=\cV^{\lambda}(v)\}=\{u: \cV(u)=\cV(v)\}$. Consider nodes $w$ and $w'$ for which
$\cV^{D+\lambda}(w)=\cV^{D+\lambda}(w')$. Let $p$ be a shortest path (coded as a sequence of ports) from $w$ to $v$. Let $v'$ be the node at the end of the same
path $p$ starting at $w'$.  $\cV^{\lambda}(v)$ is included in $\cV^{D+\lambda}(w)$ and $\cV^{\lambda}(v')$ is included in $\cV^{D+\lambda}(w')$.
It follows that $\cV^{\lambda}(v)=\cV^{\lambda}(v')$ and hence $\cV(v)=\cV(v')$. Suppose for contradiction that $\cV(w)\neq\cV(w')$. Let $q$ and $q'$ be 
paths in views $\cV(w)$ and $\cV(w')$ witnessing to their difference. Let $\overline{p}$ be the reverse of path $p$.  
Then the concatenations $\overline{p}q$ and  $\overline{p}q'$ are paths witnessing to the difference of $\cV(v)$ and $\cV(v')$, which gives a contradiction. 
\end{proof}

We fix a canonical linear order on all finite rooted trees with unlabeled nodes and labeled ports, e.g., as the lexicographic order of DFS traversals of these
trees, starting from the root and exploring children of a node in increasing order of ports.  For any subset of this class, the term ``smallest'' refers to this order.
Since views at a depth $t$ are such trees, we will elect as leader a node whose view at some depth is the smallest in some class of views. The difficulty is to establish when views at some depth are already unique for a solvable graph, and to decide fast if the graph is solvable, in the case of strong LE.

All our algorithms are written for a node $u$ of the graph.

\section{Weak leader election}

In this section we show that the optimal time of weak leader election is $\Theta(D+\lambda)$, if either the diameter $D$ of the graph or its size $n$ is known to the nodes. We first give two algorithms, one working for the class of graphs of given diameter $D$ and the other for the class of graphs of given size $n$, that elect a leader in time $O(D+\lambda)$ on graphs with diameter $D$ and level of symmetry $\lambda$, whenever election is possible. 

{ Our algorithms use the subroutine $COM$ to exchange views at different depths with their neighbors. This subroutine is detailed in Algorithm~\ref{alg:COM}.}

\begin{algorithm}
{ \caption{$COM(i)$\label{alg:COM}}

{\bf let} $\cV^i(u)$ be the view of node $u$ at depth $i$\\
{\bf send} $\cV^i(u)$ to all neighbors;\\
{\bf foreach} neighbor $v$ of $u$\\
\hspace*{1cm} {\bf receive} $\cV^i(v)$ from $v$
}
\end{algorithm}

Algorithm~\ref{alg:knownDiameter} works for the class of graphs of given diameter $D$.

\begin{algorithm*}

\caption{WLE-known-diameter($D$)\label{alg:knownDiameter}}
{\bf for} $i:=0$ {\bf to} $D-1$ {\bf do} $COM(i)$\\
compute $|\Pi_0|$; $j:=0$\\
{\bf repeat}\\
\hspace*{1cm}$COM(D+j)$; $j:=j+1$; {compute $|\Pi_j|$}\\
{\bf until} $|\Pi_j| = |\Pi_{j-1}|$\\
$V:=$ the set of nodes $v$ in $\cV^{D+j}(u)$ having the smallest $\cV^{j-1}(v)$\\
elect as leader the node in $V$ having the lexicographically smallest path from $u$

\end{algorithm*}

%
%
%
%
%
%

\begin{theorem}\label{weak-known-diam}
{Algorithm~\ref{alg:knownDiameter}  - WLE-known-diameter($D$) -} elects a leader in every solvable graph of diameter $D$, in time $O(D+\lambda)$,
where $\lambda$ is the level of symmetry of the graph.
\end{theorem}

\begin{proof}
All nodes of the graph find $|\Pi_0|$ {(i.e., the number of different node degrees in the graph)} after $D$ rounds, and then they find sizes of consecutive partitions $\Pi_j$, for $j=1, \dots , \Lambda +1$ {(by counting the number of distinct views at depth $j$)}.
At this time the exit condition of the ``repeat'' loop is satisfied. 
All nodes stop simultaneously and elect a leader. Since the graph is solvable, by the definition of $\Lambda$ and in view of Propositions \ref{possible} and  {\ref{stopBySize}}, all elements of the partition $\Pi_{\Lambda}$ are singletons {(recall that equivalence classes in $\Pi$, and thus in $\Pi_{\Lambda}$, are of equal size $\sigma$, and $\sigma=1$ in solvable graphs)}. Hence all nodes in $V$ correspond to the same node in the graph and consequently all nodes elect as leader the same node.
All nodes stop in round $D+\Lambda$, which is at most $2D+ \lambda$ by Proposition \ref{Lambda}.
\end{proof}

Algorithm~\ref{alg:knownSize} works for the class of graphs of given size $n$.

\begin{algorithm*}
\caption{WLE-known-size$(n)$\label{alg:knownSize}}

$i:=0$; $x:=1$;\\
{\bf while} $x<n$ {\bf do}\\
\hspace*{1cm}$COM(i)$; $i:=i+1$\\
\hspace*{1cm}{\bf for} $j:=0$ {\bf to} $i$ {\bf do}\\
\hspace*{2cm}$L_j:=$ the set of nodes in $\cV^{i}(u)$ at distance at most $j$ from $u$ (including $u$)\\
\hspace*{2cm}$num_j:=$ the number of nodes in $L_j$ with distinct views at depth $i-j$\\
\hspace*{1cm}$x:=$ max $\{num_j: j\in\{0,1,\dots , i\}\}$\\
compute $\Lambda$ and $D$\\
$V:=$ the set of nodes $v$ in $\cV^{i}(u)$ having the smallest $\cV^{\Lambda}(v)$\\
elect as leader the node in $V$ having the lexicographically smallest path from $u$\\
{\bf while} $i\leq D+\Lambda$ {\bf do}\\
\hspace*{1cm}
$COM(i)$; $i:=i+1$ {//Allow late nodes to terminate their protocol!}
\end{algorithm*}

%
%
%
%
%
%
%

\begin{theorem}\label{weak-known-size}
{Algorithm~\ref{alg:knownSize} - WLE-known-size$(n)$ -} elects a leader in every solvable graph of size $n$, in time $O(D+\lambda)$,
where $\lambda$ is the level of symmetry of the graph.
\end{theorem}

\begin{proof}
Consider a node $u$. After $i \leq D+\Lambda$ rounds, node $u$ gets a view  $\cV^{i}(u)$ that contains views $\cV^{\Lambda}(v )$, for all nodes $v$.
Since the graph is solvable, all views $\cV^{\Lambda}(v )$ are different, and hence $u$ exits the first ``while'' loop after seeing $n$ different views at this depth. At this point $u$ can reconstruct an isomorphic
copy of the graph and hence compute $D$ and $\Lambda$. 
The election rule is as in {Algorithm~\ref{alg:knownDiameter}}. The second ``while'' loop guarantees that every node $v$ will be able to get  $\cV^{D+\Lambda}(v)$
and hence will exit the first while loop and elect the same leader.
\end{proof}


In order to show that {Algorithm~\ref{alg:knownDiameter}} is optimal for the class of graphs of diameter $D$ 
and {Algorithm~\ref{alg:knownSize}}  is optimal for the class of graphs of size $n$  we prove the  following theorem.
It shows a stronger property: both above algorithms have optimal complexity even among weak LE algorithms working only when all three parameters $n$, $D$ and $\lambda$ are known. 

\begin{theorem}\label{lower-weak}
For any $D\geq1$ and any $\lambda \geq 0$, with $(D,\lambda)\neq (1,0)$, there exists an integer $n$ and a solvable graph $G$ of size $n$, diameter $D$ and level of symmetry $\lambda$,
such that every algorithm for weak LE working for the class of graphs of size $n$, diameter $D$ and level of symmetry $\lambda$ takes time at least $D+\lambda$ on the graph $G$\footnote{Notice that there is no solvable graph with $D=1$ and $\lambda =0$, because the latter condition, for solvable graphs, means that there is a node of a unique degree, {contradicting the requirement $D=1$, i.e., having a clique of at least 2 nodes}.}
\end{theorem}

Before proving Theorem~\ref{lower-weak} we present a construction of a family $Q_k$ of complete graphs (cliques) that will be used in the proof of our lower bound.
The construction consists in assigning port numbers.
In order to facilitate subsequent analysis we also assign labels to nodes of the constructed cliques.
{ We will use induction to describe the construction of $Q_k$.}

\begin{figure}
\begin{center}
\begin{tabular}{cc}
(a)&

\scalebox{0.4}{
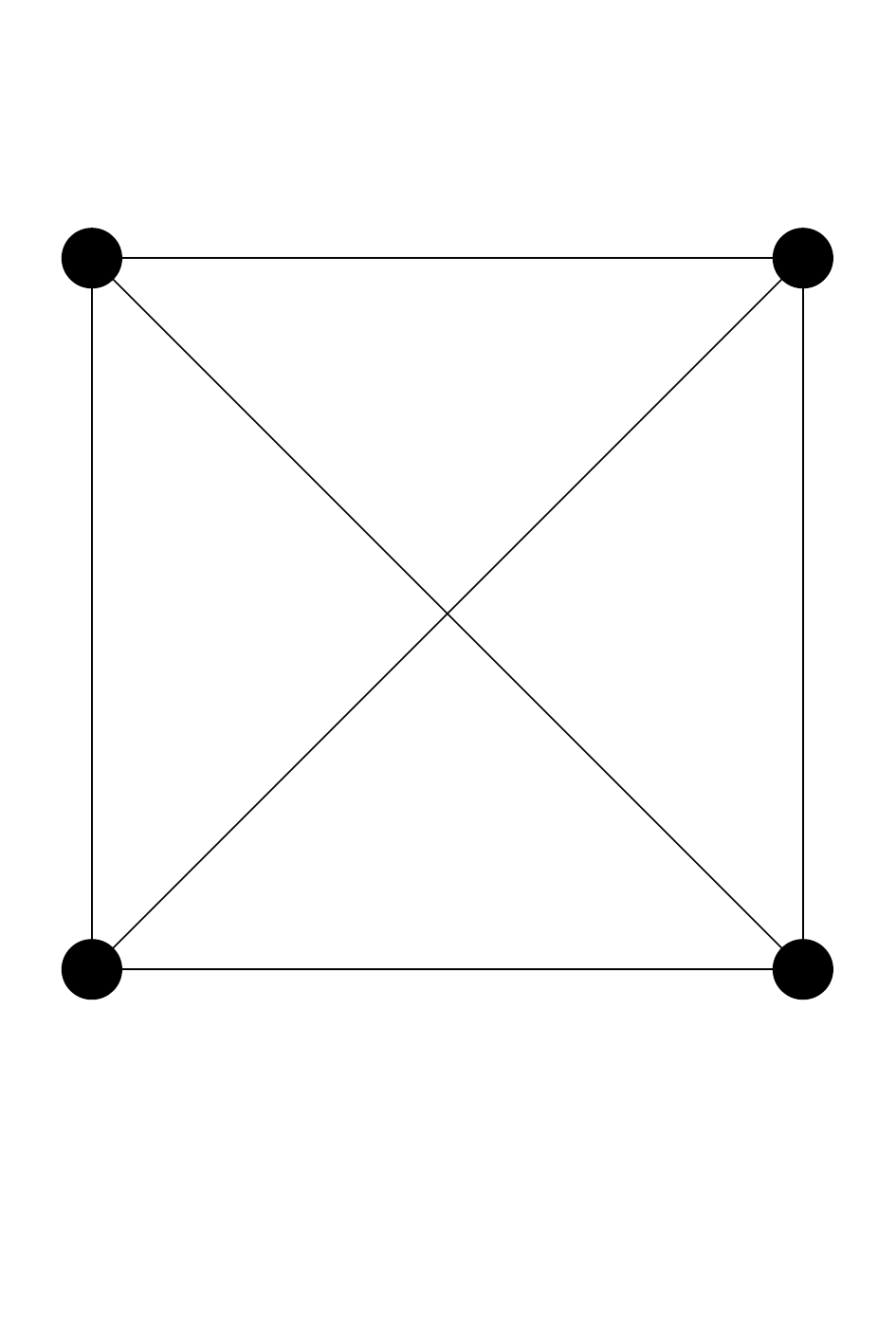}\\\\\hline\\
(b)&
\scalebox{0.4}{
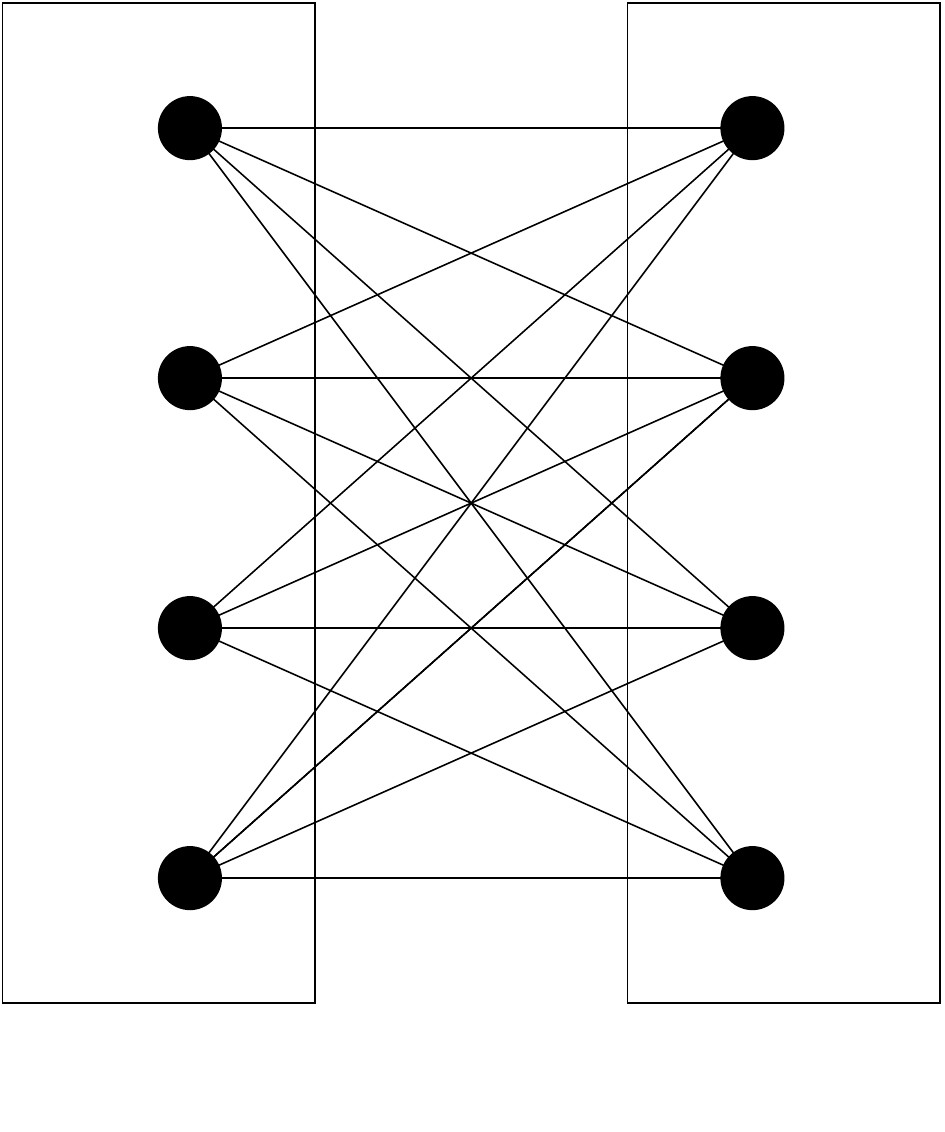}\\\\\hline\\
(c)&
\scalebox{0.4}{
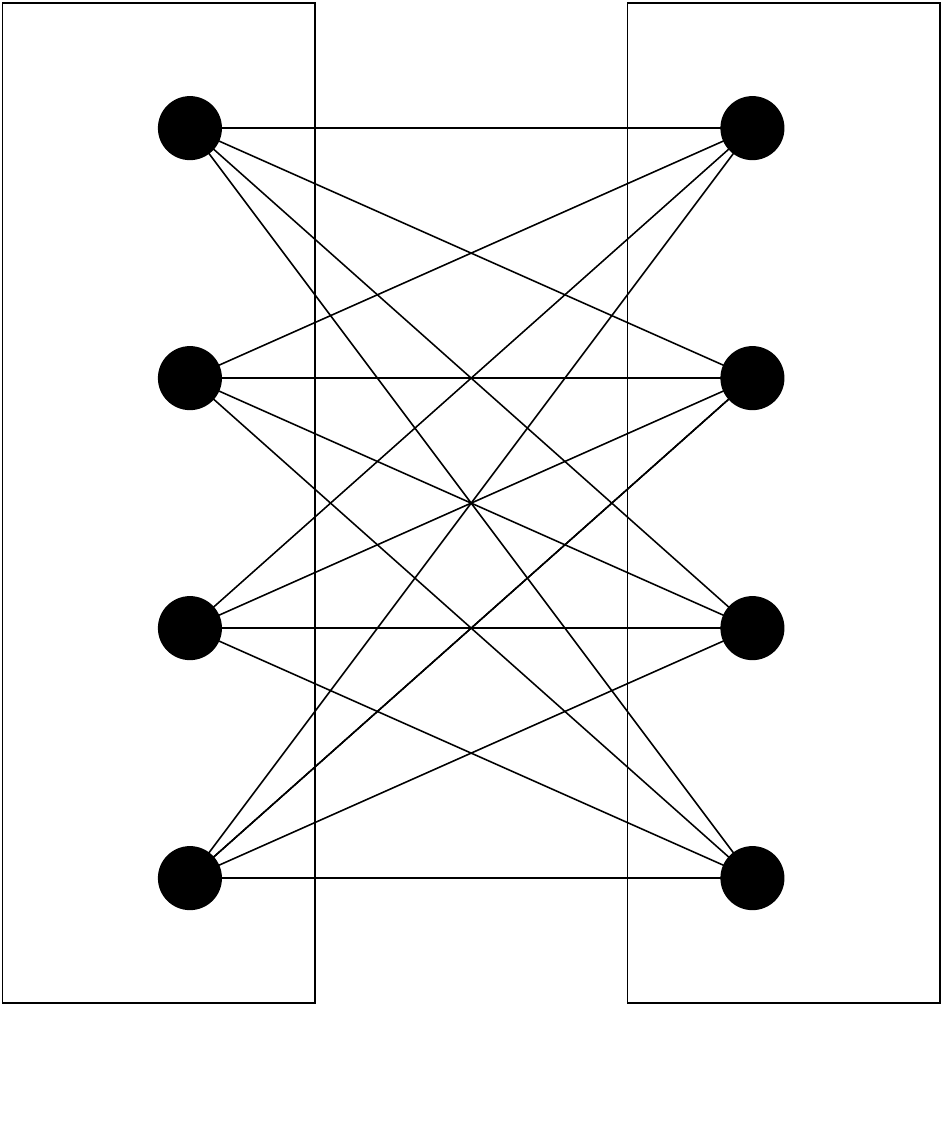}
\\\\\hline\\
(d)&
\scalebox{0.4}{
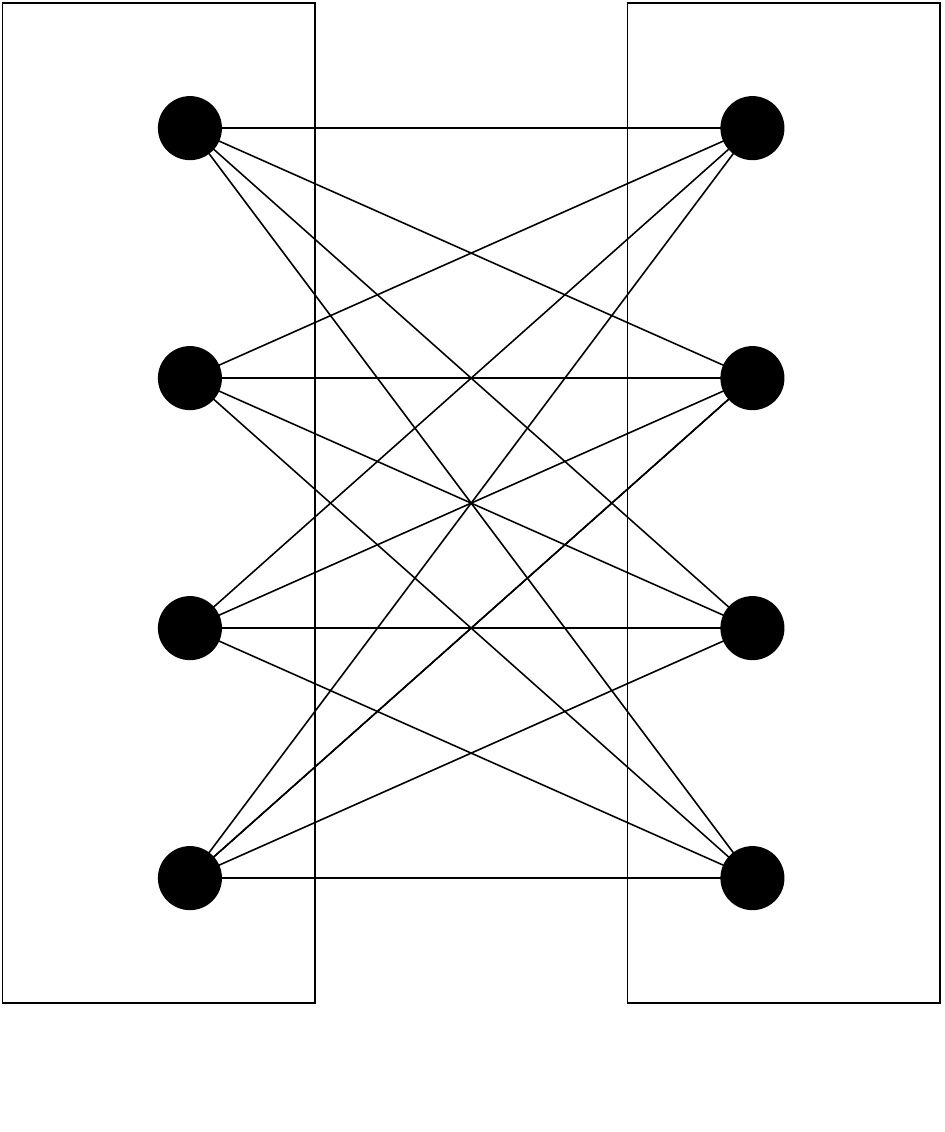}
\end{tabular}
\caption{\label{fig.Q2Q3} a) Clique $Q_2$. b) Clique $Q_3$  {(internal edges of $Q_2$ and $\overline{Q}_2$ are omitted).
c) and d) Extracts of graph $Q_4$.}}

\end{center}
\end{figure}

%

$Q_1$ is the single node graph. $Q_2$ is a 4-node clique with port numbers defined as follows.
For a node $u$, we say that edge $(i,j)$ is \emph{incident} to $u$, if the edge corresponding to port $i$ at node $u$ corresponds to port $j$ at some node $v$.
Nodes in $Q_2$ are uniquely identifiable by the set of their incident edges.
Below we assign distinct labels to the four nodes of the clique depending on the sets of their incident edges (see Fig.~\ref{fig.Q2Q3}.a):

\begin{itemize}
\item set of edges $\{(0,0)$, $(1,1)$,  $(2,2)\}$ -- corresponding to label $a$;
\item set of edges $\{(0,0)$, $(1,1)$,  $(2,0)\}$ -- corresponding to label $b$;
\item set of edges $\{(0,2)$, $(1,1)$,  $(2,0)\}$ -- corresponding to label $c$;
\item set of edges $\{(0,2)$, $(1,1)$,  $(2,2)\}$ -- corresponding to label $d$.
\end{itemize}

We additionally assign colors 0, 1, and 2 to edges of $Q_2$ as follows:
edges $\{a,b\}$ and $\{c,d\}$ get color 0,
edges $\{a,c\}$ and $\{b,d\}$ get color 1, and 
edges $\{a,d\}$ and $\{b,c\}$ get color 2.

$Q_3$ is constructed starting from two disjoint copies of clique $Q_2$ as follows.
Denote by $Q_2$ one of the copies and by $\overline{Q}_2$ the other one. 
{For each node $x$ in $Q_2$, we denote by $\overline{x}$ its corresponding node in $\overline{Q}_2$.

Ports 3, 4, 5, and 6 are used to connect each node in $Q_2$ to all nodes in $\overline{Q}_2$ to construct an 8-node clique.
We will use two types of edges. {\em Monochromatic edges} 
will have port 3 or 4 at both endpoints, while {\em skew edges} will have port 5 at one endpoint and port 6 at the other endpoint. For monochromatic
edges we call {\em color} the port number at both their endpoints. We will also consider as monochromatic the edges in $Q_2$.


More precisely, edges of color 3 connect nodes
$\{a, \overline{c}\}$, $\{\overline{a},c\}$, $\{b, \overline{d}\}$, and $\{\overline{b}, d\}$.
Edges of color 4 connect nodes 
$\{a, \overline{d}\}$, $\{\overline{a},d\}$, $\{b, \overline{c}\}$, and $\{\overline{b}, c\}$.


Notice that until now, this construction results in a graph
for which each node has exactly one
other node with the same view, e.g., $a$ and $\overline{a}$ have the same view.
Uniqueness of all views at depth 2 is guaranteed in clique $Q_3$ by the definitions of the skew edges.

In particular, skew edges connect nodes:
\begin{itemize}
\item $\{a, \overline{a}\}$, $\{b, \overline{b}\}$, $\{c, \overline{c}\}$, and $\{d, \overline{d}\}$ with port 6 at nodes $a,b,c$, and $d$ and port 5 at nodes $\overline{a}, \overline{b}, \overline{c}$, and $\overline{d}$;
\item $\{\overline{a},b\}$, $\{\overline{b},a\}$, $\{\overline{c},d\}$, and $\{\overline{d},c\}$ with port 6 at nodes $\overline{a}, \overline{b}, \overline{c}$, and $\overline{d}$ and port 5 at nodes $a,b,c$, and $d$.
\end{itemize}

 This concludes the construction of $Q_3$ (see Fig.~\ref{fig.Q2Q3}.b).
 
A node $x$ in $Q_3$  whose outgoing port 6 leads to a node $y$ receives as its label the concatenation of the labels of nodes $x$ and $y$ in their respective copies of $Q_2$ (removing all overlines).

The complete labeling of nodes in $Q_3$ is:
$aa$, $ab$, $bb$, $ba$, $cc$, $cd$, $dd$, and $dc$.

For $k \ge 3$, the clique $Q_{k+1}$ is produced starting from disjoint copies $Q_k$ and $\overline{Q}_k$ of clique $Q_k$ as follows.

The set of nodes of the clique $Q_{k+1}$ is the union of the sets of nodes of $Q_k$ and $\overline{Q}_k$.
The skew edges connecting nodes of $Q_k$ and $\overline{Q}_k$ will have port numbers $(2^{k+1}-2, 2^{k+1}-3)$, while the monochromatic edges will use the same port $i$, for $i\in[2^k-1, 2^{k+1}-4]$, on both endpoints.
We denote by $\alpha_{j}$ a string of length $2^j$ over the alphabet $\{a,b,c,d\}$. For $h\le j$, $\alpha_h$ is the prefix of length $2^h$ of string $\alpha_j$.
We assign port numbers to the edges according to the following two rules.

\begin{enumerate}
\item {\bf Skew edges:} 
the edge that has port $2^{k+1}-2$ at a node $\alpha_{k-3}\beta_{k-3}$, has port $2^{k+1}-3$ at its other endpoint $\overline{\alpha_{k-3}\beta_{k-3}}$;
the edge that has port $2^{k+1}-2$ at a node $\overline{\alpha_{k-3}\beta_{k-3}}$, has port $2^{k+1}-3$ at its other endpoint $\alpha_{k-3}\gamma_{k-3}$ (where $\beta_{k-3}\neq \gamma_{k-3})$.

\item {\bf Monochromatic edges:} let node $\alpha_{k-3}\beta_{k-3}$ and node $\overline{\gamma_{k-3}\delta_{k-3}}$ 
be connected by a monochromatic edge of color $i$; then
node $\overline{\alpha_{k-3}\beta_{k-3}}$ and node \\$\gamma_{k-3}\delta_{k-3}$ 
are connected by a monochromatic edge of color $i$. Moreover, for each $j\in[0,k-3]$, any node $\alpha_j\eta$ is connected to some node $\overline{\gamma_j\zeta}$ by a monochromatic edge of color $i$.
\end{enumerate}

The label of a node $v$ in $Q_{k+1}$ is given by the concatenation of the labels of the node in $Q_k$ (respectively in $\overline{Q}_k$)  corresponding to $v$ and its neighbor in $\overline{Q}_k$ (respectively $Q_k$) connected by the skew edge $(2^{k+1}-2,2^{k+1}-3)$ incident to $v$.

While rule 1 is constructive, rule 2 is not. However, a port assignment for monochromatic edges of $Q_{k+1}$ that fulfills rule 2, can be obtained, for $k\ge 3$, by exploiting the already defined edges of $Q_k$ as follows.
If there is a monochromatic edge with color $i$ between nodes $\{u,v\}$ in $Q_k$, then
nodes $\{u,\overline{v}\}$ and nodes $\{\overline{u},v\}$ are connected by a monochromatic edge with color $2^{k}+i-1$.
If $\{u, v\}$ are connected by the skew edge $(2^{k}-3,2^{k}-2)$ incident to $u$, then nodes $\{u,\overline{v}\}$ and nodes $\{\overline{u},v\}$ are connected by monochromatic edges of color $2^{k+1} -4$.
Now consider the skew edges from previous steps of the construction i.e., skew edges with port numbers $(2^{j}-2, 2^{j}-3)$
of graph $Q_j$, with $j\le k-1$).
Let $\{u,\overline{u}, v, \overline{v}\}$ be four nodes connected by these edges, where $\{u, v\}$ are nodes from copy $Q_{j-1}$ in the construction of $Q_j$, and nodes $\{\overline{u},\overline{v}\}$ are the corresponding nodes in copy $\overline{Q}_{j-1}$.
More precisely, let:
\begin{itemize}
\item edge $(2^{j}-2, 2^{j}-3)$ incident to $u$ connect nodes $u$ and $\overline{u}$;
\item edge $(2^{j}-2, 2^{j}-3)$ incident to $v$ connect nodes $v$ and $\overline{v}$;
\item edge $(2^{j}-2, 2^{j}-3)$ incident to $\overline{u}$ connect nodes $\overline{u}$ and $v$;
\item edge $(2^{j}-2, 2^{j}-3)$ incident to $\overline{v}$ connect nodes $\overline{v}$ and $u$.
\end{itemize}
Finally, let $\overline{u}$ and $\overline{v}$ from $Q_j$ correspond to nodes $w$ and $z$, respectively, in copy $Q_k$ of the construction of $Q_{k+1}$.
Then nodes $\{u, \overline{w}\}$, $\{\overline{u}, w\}$, $\{v, \overline{z}\}$, $\{\overline{v}, z\}$ are connected by  monochromatic edges of color $2^{k}+2^{j}-4$;
nodes $\{u, \overline{z}\}$, $\{\overline{u}, z\}$, $\{v, \overline{w}\}$, $\{\overline{v}, w\}$ are connected by  monochromatic edges of color $2^{k}+2^{j}-3$. This concludes the definition of monochromatic edges of $Q_{k+1}$, and thus completes its construction.
Extracts from graph $Q_4$ are depicted in Fig.~\ref{fig.Q2Q3}.c and~\ref{fig.Q2Q3}.d.

%
%
%
%
}

A node of the clique $Q_k$ is said to be of \emph{type} $a$, $b$, $c$, or $d$ if it is obtained from a node with label $a$, $b$, $c$, or $d$ (respectively) in a copy of $Q_2$ in the construction of $Q_k$. 
{Notice that the type of a node corresponds to the first letter of its label.}
Consider a path $p$ defined as a sequence of consecutive monochromatic and skew edges (i.e., $p$ is uniquely defined by the sequence of outgoing port numbers).
{If a skew edge $e$ is traversed from node $u$ to $v$ by $p$, then the {\em color} of $e$ in $p$ is given by its port number  at $u$.}
We call $p$ a \emph{distinguishing} path for nodes $x$ and $y$ in $Q_k$, if it yields  two different sequences of node types $a$, $b$, $c$, and $d$ traversed proceeding along $p$, depending on whether the origin of $p$ is $x$ or $y$.

\begin{lemma}\label{clique}
{The clique $Q_{k+1}$ has level of symmetry $k$.}
\end{lemma}

\begin{proof}
{
We will prove a stronger statement. We will show that, for $k\ge 2$,
if  $j \in [0, k-2]$ 
and two nodes $u$ and $v$ in $Q_{k+1}$ are such that their labels have a common prefix $\alpha_j$ but no common prefix $\alpha_{j+1}$, then the views of $u$ and $v$ are identical up to depth $j+1$ and differ at depth $j+2$.

The base step of the inductive proof is for $k=2$. Indeed, node $aa$ in $Q_3$ has the same view as node $ab$ at depth $1$, while their views differ at depth $2$. The same happens for the views of nodes $bb$ and $ba$, $cc$ and $cd$, $dd$ and $dc$. The views of any other pair of nodes are different at depth 1.

Consider the construction of clique $Q_{k+1}$. Let $j \in[0, k-3]$ and consider two nodes $u$ and $v$ whose labels have a common prefix $\alpha_j$ but no common prefix $\alpha_{j+1}$.
Let $i\in[2^k-1,2^{k+1}-2]$, and consider the neighbors $u'$ and $v'$ (of $u$ and $v$ respectively) that are the other endpoints of the edges having port $i$ at $u$ and $v$. The labels of $u'$ and $v'$ have a common prefix $\beta_j$.
Indeed, if port $i$ is a port number of a monochromatic edge, 
rule 2 immediately implies that $u'$ and $v'$ have a common prefix $\beta_j$.
On the other hand, if $i$ is a port number of a skew edge then the first half of the label of $u'$  must coincide with that of $u$, and the first half of the label of $v'$ with that of $v$ in graph $Q_k$, which in turn implies that $u'$ and $v'$ have a common prefix $\beta_j=\alpha_j$ in graph $Q_{k+1}$.
It follows that edges with colors in $[2^k-1,2^{k+1}-2]$ cannot appear in any distinguishing path of minimal length between two nodes $u$ and $v$ whose labels have a common prefix of length up to $2^{k-3}$, and therefore the depth at which views of such nodes differ in $Q_{k+1}$ is the same as that of their corresponding nodes in $Q_k$.

It remains to consider the case of nodes $u\in Q_k$ and $\overline{u}\in \overline{Q}_k$, sharing a label prefix $\alpha_{k-2}$ in $Q_{k+1}$.
For any such pair of nodes, if $i\in [2^k-1,2^{k+1}-4]$ is the color of a monochromatic edge connecting $u$ to $\overline{v}$, then another monochromatic edge of color $i$ connects $\overline{u}$ to $v$. Any pair of such nodes have the same view (at any depth) in $Q_k$.
On the other hand, the skew edge $(2^{k+1}-2, 2^{k+1}-3)$ incident to $u$ and that incident to $\overline{u}$ lead, by rule 1, to two nodes having a common label prefix of length $\alpha_{k-3}$, but no common prefix $\alpha_{k-2}$
(the same reasoning applies to the skew edges $(2^{k+1}-3,2^{k+1}-2)$ incident to $\overline{u}$ and $u$). 
These nodes have identical views up to depth $k-2$ in $Q_k$, but distinct at depth $k-1$, by the inductive hypothesis. Hence $u$ and $\overline{u}$ have the same view up to depth $k-1$ and distinct views at depth $k$ in $Q_{k+1}$, which concludes the proof.
}
\end{proof}

\begin{remark} 
After the publication of the conference version of this paper
a more general version of Lemma~\ref{clique} has been proved in \cite{DKP13}:
for any $D \leq n$ there exists a $\Theta(n)$-node graph of diameter $\Theta(D)$
with level of symmetry $\Omega(D\log(n/D))$.
\end{remark}

We will also use the following family of cliques $\widetilde{Q}_k$.
$\widetilde{Q}_1$ is the clique on 2 nodes, with port number 0.
$\widetilde{Q}_2$ is a clique on 4 nodes, where all nodes have the same set of incident edges $\{(0,0),(1,1),(2,2)\}$.
{For $k\ge2$,  $\widetilde{Q}_{k+1}$ is a clique obtained from two disjoint copies of $Q_{k}$.
The construction of $\widetilde{Q}_{k+1}$ mimics the construction of $Q_{k+1}$ for all edges but the skew edges between nodes in $Q_k$ and $\overline{Q}_k$, that are replaced by monochromatic edges with port number $2^{k+1}-2$ at both endpoints, connecting nodes $u$ and $\overline{u}$, and by monochromatic edges with port number $2^{k+1}-3$ at both endpoints, connecting nodes $\alpha_{k-3}\alpha_{k-3}$ and $\overline{\alpha_{k-3}\beta_{k-3}}$ and nodes $\overline{\alpha_{k-3}\alpha_{k-3}}$ and $\alpha_{k-3}\beta_{k-3}$.}
Notice that, in graph $\widetilde{Q}_k$, nodes $x$ and $\overline{x}$ have identical views. Nevertheless, in order to describe our construction, we artificially assign to nodes $x$ and $\overline{x}$ the label they would respectively receive in the construction of graph $Q_{k}$.

We finally define a family of graphs that allow us to prove Theorem~\ref{lower-weak}.
For any pair of integers $(D, \lambda)$, with $D\ge2$  and $\lambda \ge 2$, the graph $R_{D,\lambda}$ is obtained using one copy of graph $Q_{\lambda+1}$ and $2D-1$ copies of graph $\widetilde{Q}_{\lambda +1}$. The construction of 
 graph $R_{D,\lambda}$ proceeds as follows.
Arrange $2D-1$ disjoint copies of $\widetilde{Q}_{\lambda +1}$ and one copy of $Q_{\lambda+1}$ in a cyclic order. Connect each node in a clique with all nodes in the subsequent clique.
{ 
Let $\{x, y\}$ be two nodes in $Q_{\lambda +1}$ 
and let $i$ be the color that would be assigned to edge edge $(x,\overline{y})$ in the construction of $\widetilde{Q}_{\lambda +2}$.
Assign port numbers $(i, i+2^{\lambda +1})$ to the edge connecting node $x'$ in some clique
to node $y''$ in the subsequent clique, where $x'$ has label $x$ and $y''$ has label $y$.}

A distinguishing path in $R_{D,\lambda}$ is defined in the same way as in $Q_{\lambda +1}$, which is possible, since in both graphs each node has type $a$, $b$, $c$, or $d$.

\begin{lemma}\label{ring}
Let $\lambda\ge2$ and $D\ge2$.
Let $x$ and $y$ be two nodes in $Q_{\lambda+1}$ and let $\ell \le \lambda -1$ be the maximum depth at which views of $x$ and $y$ in $Q_{\lambda+1}$
are identical.
Then the views at depth $\ell$  of nodes $x$ and $y$ belonging to the copy of $Q_{\lambda +1}$ in $R_{D, \lambda}$ are identical.
\end{lemma}

\begin{proof}
{
It follows from the definition that the length of a shortest distinguishing path for nodes $x$ and $y$ in $Q_{\lambda +1}$  is $\ell$.
By construction of $Q_{\lambda+1}$, nodes $x$ and $y$ have labels with an identical prefix of length $2^{\ell +1}$.
Suppose, for contradiction, that  views of $x$ and $y$ at depth $\ell$ are different in $R_{D, \lambda}$.
Hence a shortest distinguishing path for these nodes has length $t<\ell$ in $R_{D, \lambda}$. Let $p$ be a distinguishing path of length $t$ in  $R_{D, \lambda}$ for nodes $x$ and $y$.
We will show how to construct a distinguishing path $p'$ of length at most $t$ in $Q_{\lambda+1}$, which will give a contradiction.


Edge colors $\{0,1,2\}$ inside each clique $\widetilde{Q}_{\lambda+1}$ are assigned according to the rules used for cliques $Q_{\lambda +1}$.
Monochromatic edges with port number $i$ at both endpoints get color $i$.
For edges connecting nodes from different cliques in $R_{D,\lambda}$ we assign as color the smaller of their port numbers.
(Notice that edges connecting nodes with the same label in their cliques will get the same color $2^{\lambda +2}-2$ but these edges will be subsequently deleted in the construction of the distinguishing path.) 
For the remaining edges (i.e., skew edges in the construction of a graph $Q_k$, for some $k\le \lambda$), the color is defined by the outgoing port number, according to path $p$.

The distinguishing path $p'$ is constructed as follows.
Consider all edges $e$ in $p$, in reverse order.
If  $e$ has color $h= 2^{\lambda +2}-2$ then we remove it; if $e$ has color $h$ such that $2^{\lambda +2}-3 \geq h \geq 2^{\lambda +1}-1$ then we replace it with an edge with color $h-2^{\lambda +1}$.
This corresponds to replacing edges going from one copy of $\widetilde{Q}_{\lambda+1}$ or from $Q_{\lambda+1}$ to another copy of $\widetilde{Q}_{\lambda+1}$ or to $Q_{\lambda}$ with internal edges of some clique $\widetilde{Q}_{\lambda+1}$ or of $Q_{\lambda +1}$. Deleted edges are those going to corresponding nodes of different cliques, hence their deletion corresponds to removing self loops.
%
%
We show that path $p'$ is distinguishing for nodes $x$ and $y$ in $Q_{\lambda +1}$.

Indeed, the only edge replacements that could modify the sequence of $a$, $b$, $c$, and $d$ types yielded by paths $p$ and $p'$, when starting from nodes $x$ and $y$, are those of edges with color and $2^{\lambda +2}-3$, as these edges may be replaced by skew edges of the last step of the construction of the clique $Q_{\lambda +1}$, that are defined differently in cliques $\widetilde{Q}_{\lambda +1}$.
For each such edge leading to a node with label $\alpha_{\lambda -2}\beta_{\lambda -2}$, the corresponding skew edge in $Q_{\lambda +1}$ leads to a node with label $\alpha_{\lambda -2}\gamma_{\lambda -2}$.
Hence, views in $Q_{\lambda +1}$ of nodes $\alpha_{\lambda -2}\beta_{\lambda -2}$ and $\alpha_{\lambda -2}\gamma_{\lambda -2}$ are identical up to depth $\lambda -1$, as shown in the proof of Lemma~\ref{clique}, and thus the path $p'$, of length at most $t$, is distinguishing for $x$ and $y$ in $Q_{\lambda +1}$, contradiction.}
\end{proof}

{\bf Proof of Theorem~\ref{lower-weak}:}
{ Lemma \ref{small} below proves the theorem if either $D$ or $\lambda$ are less than 2. 
Here we give the general argument for
$D, \lambda \ge 2$. 
Consider the clique $\widetilde{Q}_{\lambda + 1}$ antipodal to the  clique $Q_{\lambda +1}$ in graph $R_{D,\lambda}$.
Consider nodes $x$ and $\overline{x}$ from this clique.
Any distinguishing path for nodes $x$ and $\overline{x}$ in $R_{D,\lambda}$ must contain a node from $Q_{\lambda +1}$.
Let $q$ be a minimum length distinguishing path for nodes $x$ and $\overline{x}$ and assume without loss of generality that $y$ and $\overline{y}$ are the first nodes from $Q_{\lambda+1}$ found along path $q$, if starting from $x$ and $\overline{x}$, respectively.
By Lemmas \ref{clique} and \ref{ring}, nodes $y$ and $\overline{y}$ have the same views at depth $\lambda -1$ and different views at depth $\lambda$
in the graph $R_{D,\lambda}$. Thus the minimum length distinguishing path in $R_{D,\lambda}$ for $y$ and $\overline{y}$ has length $\lambda - 1$.
Since nodes $y$ and $\overline{y}$ are at distance $D$ from $x$ and $\overline{x}$, respectively, the views at depth $D+\lambda -1$ of $x$ and $\overline{x}$
are identical.

 The following lemma proves Theorem \ref{lower-weak} in the case when either $D$ or $\lambda$ are small, thus concluding the proof of Theorem~\ref{lower-weak}.}

\begin{lemma}\label{small}
For $D=1$ and any $\lambda \geq 1$, and for any $D \geq 2$ and $0 \leq \lambda \leq 1$, 
there exists an integer $n$ and a solvable graph $G$ of size $n$, diameter $D$ and level of symmetry $\lambda$,
such that every algorithm for weak LE, working for the class of graphs of size $n$, diameter $D$ and level of symmetry $\lambda$, takes time at least $D+\lambda$
on the graph $G$.
\end{lemma}

\begin{proof}
Consider the following cases.

\noindent
Case 1. $\lambda=0$ and $D \geq 2$.

\begin{figure*}
\begin{center}
\scalebox{0.4}{
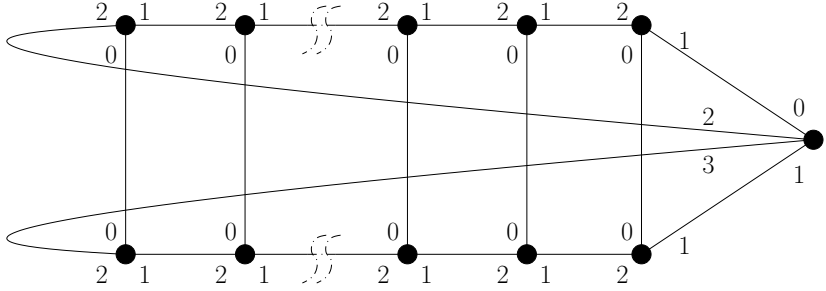}
\caption{\label{fig.lambda0} Graph $G$ of diameter $D\ge1$ and level of symmetry 0.}
\end{center}
\end{figure*}

The required graph $G$ is depicted in Fig.~\ref{fig.lambda0}. The lower bound $D$ on the time of weak LE is straightforward.

\noindent
Case 2. $\lambda =1$ and $D \geq 2$.

The construction of the graph $G$ is analogous to the general construction of $R_{D,\lambda}$. We take one copy of clique $Q_2$ and $2D-1$ copies
of clique  $\widetilde{Q}_2$ in circular order. In copies of cliques $\widetilde{Q}_2$ we arbitrarily assign labels $a$, $b$, $c$, and $d$ to nodes. The argument
remains the same as in the general case. 

\noindent
Case 3. $\lambda \geq 1$ and $D= 1$.

Consider the clique constructed from cliques $Q_{\lambda +1}$  and $\widetilde{Q}_{\lambda + 1}$ by connecting each node with label $x$ in $Q_{\lambda +1}$
to each node with label $y$ in $\widetilde{Q}_{\lambda + 1}$ by an edge 
having the color that would be assigned to edge $(x,\overline{y})$ in the construction of $\widetilde{Q}_{\lambda+2}$.
The obtained clique requires time $1+\lambda$ for weak LE.
\end{proof}

\section{Strong leader election}

For strong leader election more knowledge is required to accomplish it, and even more knowledge is needed to perform it fast. We first prove that knowledge of
the diameter $D$ is not sufficient for this task. The idea is to show, for sufficiently large $D$, one solvable and one non-solvable graph of diameter $D$, such that
both graphs have the same sets of views.

\begin{theorem}\label{impos}
For any $D\geq  4$ there is no strong LE algorithm working for all
graphs of diameter~$D$.
\end{theorem}

\begin{figure}
\begin{center}
\begin{tabular}{cc}
(a)&\scalebox{0.25}{
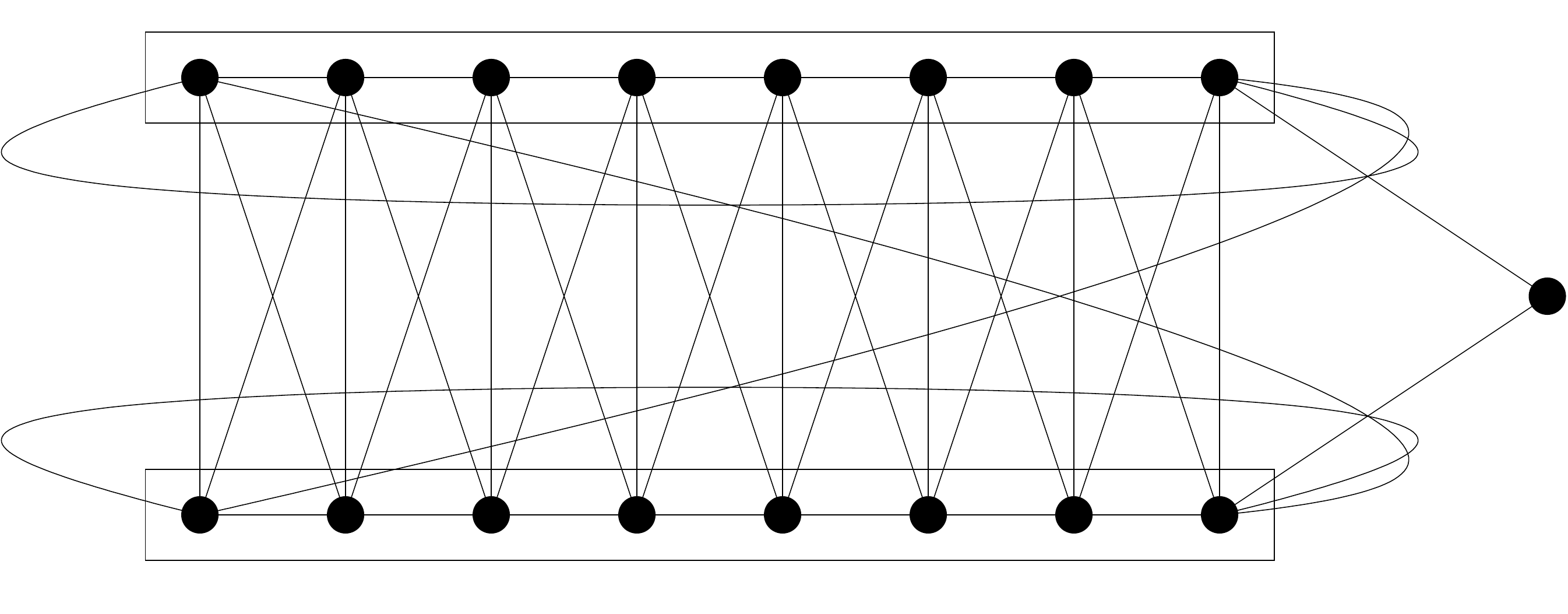}\\\\\hline\\
(b)&\scalebox{0.25}{
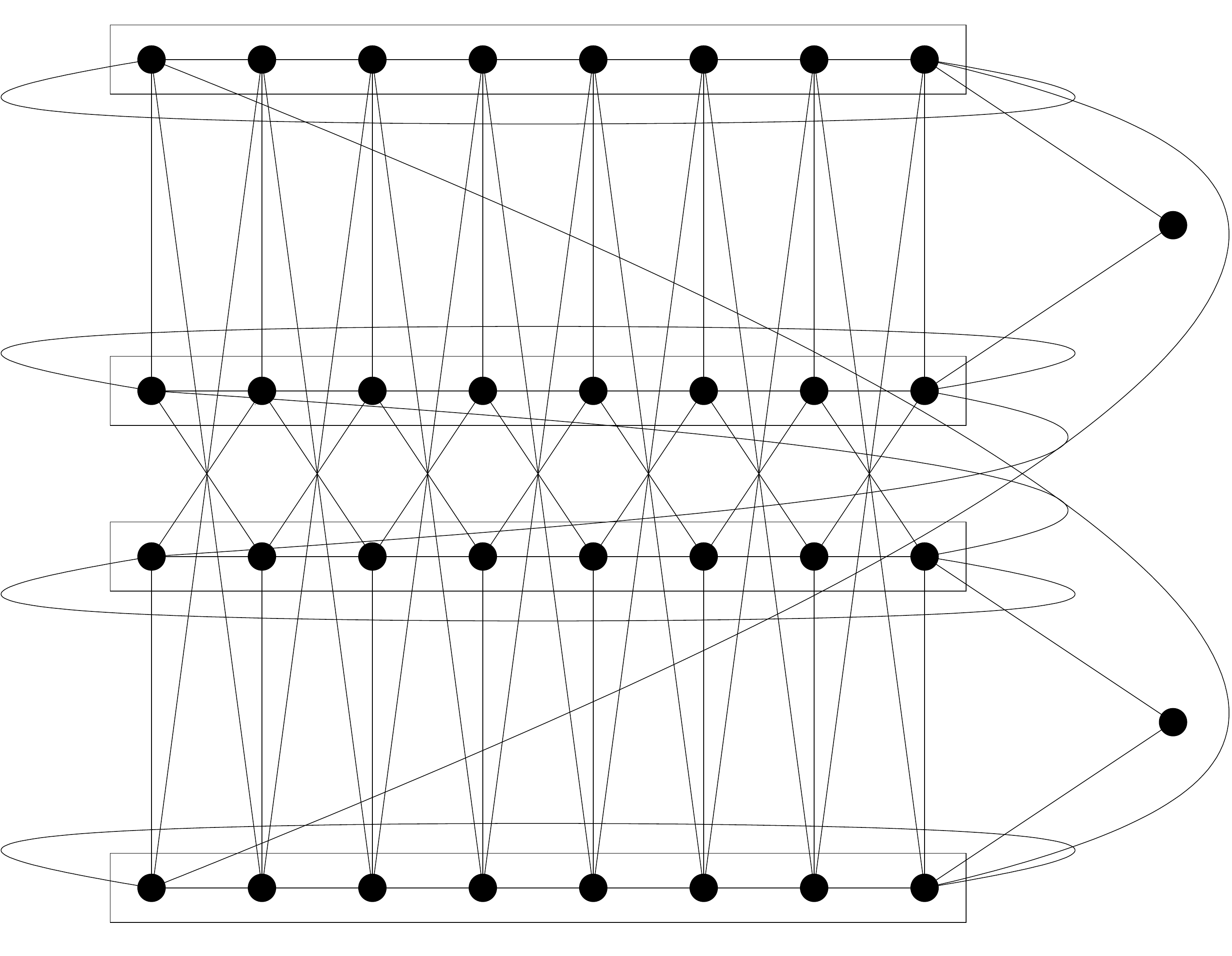}
\end{tabular}
\caption{\label{fig.T-k} a) The graph $T_4$. b) The graph $M_4$.}
\end{center}
\end{figure}

\begin{proof}
Let $k\geq 3$.
We will use the following family of graphs $T_k$ {(see also Fig.~\ref{fig.T-k}.a)}.
Consider a $2 \times 2k$ torus $\tau$.
Let $R$ and $R'$ be the two rings of size $2k$ in $\tau$.
For each node $u\in R$, let $u'$ be its unique neighbor from $R'$ in $\tau$.
The graph $T_k$ is obtained by connecting each node $u\in R$ with nodes
$v'$ and $w'$ in $R'$, where $v$ and $w$ are the two neighbors of $u$
in $R$.
An additional node $z$ of degree 2, connected to a pair of nodes $u$,
$u'$, completes the construction of $T_k$.
The assignment of port numbers can be performed arbitrarily and is thus
omitted in the construction.

Consider the following family $M_k$ of graphs {(see also Fig.~\ref{fig.T-k}.b)}.
The graph $M_k$ is obtained from two copies of the graph $T_k$. Let us call
$T_k$ one copy, and $\overline{T}_k$, the other. Similarly we call
$\overline{u}$ the copy in $\overline{T}_k$ of node $u\in T_k$.
The graph $M_k$ is obtained by removing all edges $\{u, v'\}$  and
$\{\overline{u}, \overline{v'}\}$
from $T_k$ and $\overline{T}_k$, and replacing them with edges $\{u,
\overline{v'}\}$ and $\{\overline{u}, v' \}$,
maintaining the same port numbers.

Since the graph $T_k$ has exactly one node of
degree 2, it is solvable for any assignment of port numbers.
By the construction of graph $M_k$, each node $u$ from one copy of $T_k$
has the same view as the corresponding node $\overline{u}$ from the
other copy. It follows that graph $M_k$ is not solvable.
To prove the theorem, it is thus enough to show that graphs $T_k$ and
$M_k$ have the same diameter. Indeed, this will imply that no algorithm knowing only the diameter can tell apart
some solvable graph from an unsolvable one.

The diameter of graph $T_k$ is $k+1$. Indeed, $k+1$ is the distance
from the unique node $z$  of degree two to the two nodes in rings
$R$ and $R'$ that are antipodal to the neighbors of $z$.
All other pairs of nodes are at distance at most $k$.
Consider the graph $M_k$.
Since no edge from the original torus in the copies of $T_k$ is
modified in the construction of $M_k$, we have that all pairs of nodes
from the copy $T_k$, as well as all pairs from the copy
$\overline{T}_k$ are within distance $k+1$ from each other.
By using one edge $\{u,\overline{u'}\}$ each node $u\in R$ can
reach any node $\overline{v'}$ in $\overline{R'}$ within $k$ steps
(similarly for $u'\in R'$). 
Nodes $u\in R$ and $\overline{u}\in \overline{R}$ are at distance 3  in $M_k$, for any $k\geq 3$.
The two copies of the
node of degree two are at distance 4 in $M_k$, for any 
$k \geq 3$. Hence, for any $k \geq 3$,
all pairs of nodes in $M_k$
are within distance $k+1$ from each other. 
On the other hand, the distance in $M_k$ from a node $u$ in $R$ and a node $\overline{x}$ in $\overline{R}$, where $x$ is the antipodal node of $u$ in $R$, is equal to $k+1$. 
Hence the diameter of $M_k$ is $k+1$, which concludes the proof.
\end{proof}

By contrast, knowledge of the size $n$ alone is enough to accomplish strong leader election, but (unlike for weak LE), it may be slow.
We will show that optimal time for strong LE is $\Theta(n)$ in this case. We first show Algorithm~\ref{alg:SLEknownSize}, working in time $O(n)$.

\begin{algorithm*}
\caption{SLE-known-size$(n)$\label{alg:SLEknownSize}}

{\bf for} $i:=0$ {\bf to } $2n-3$ {\bf do} $COM(i)$\\
$L:=$ the set of nodes in $\cV^{2n-2}(u)$ at distance at most $n-1$ from $u$ (including $u$)\\
$num:=$ the number of nodes in $L$ with distinct views at depth $n-1$\\
{\bf if} $num<n$ {\bf then} report ``LE impossible''\\
{\bf else}\\
\hspace*{1cm}$V:=$ the set of nodes $v$ in $\cV^{2n-2}(u)$ having the smallest $\cV^{n-1}(v)$\\
\hspace*{1cm}elect as leader the node in $V$ having the lexicographically smallest path from $u$

\end{algorithm*}
%
%
%
%
%
%
%

\begin{theorem}\label{strong-size}
{Algorithm~\ref{alg:SLEknownSize} - SLE-known-size$(n)$ -} performs strong LE in the class of graphs of size $n$, in time $O(n)$.
\end{theorem}

\begin{proof}
After $2n-2$ rounds of communication every node $u$ has view  $\cV^{2n-2}(u)$, which contains views $\cV^{n-1}(v)$ of all nodes $v$.
A copy of each node appears in the set $L$.
By Propositions \ref{possible} and \ref{trunc} the graph is solvable if and only if there are exactly $n$ different views $\cV^{n-1}(v)$ of nodes in $L$.
Hence the report ``LE impossible'' is correct when $num<n$. If LE is possible, the election rule is the same as in the weak LE algorithms. Hence all nodes elect the same
leader in this case.
\end{proof}

Our next result shows that {Algorithm~\ref{alg:SLEknownSize}} is optimal if only $n$ is known. Compared to Theorem \ref{weak-known-size} 
it shows that the optimal time of 
strong LE with known size can be exponentially slower than that of weak  LE with known size. Indeed, it shows that strong LE may take time
$\Omega(n)$ on some graphs of diameter logarithmic in their size and having level of symmetry 0,  while weak LE takes time $O(\log n)$, on any solvable graph of diameter  $O(\log n)$ and level of symmetry 0.

The high-level idea of proving that {Algorithm~\ref{alg:SLEknownSize}} is optimal if only $n$ is known is the following. For arbitrarily large $n$, we show one solvable and one non-solvable graph of size $n$, such that  there are nodes in one graph having the same view at depth $\Omega(n)$ as some nodes of the other.

\begin{theorem}\label{lower-strong}
For arbitrarily large $n$ there exist graphs $H_n$ of size $n$,  level of symmetry 0 and diameter $O(\log n)$, such that every strong LE algorithm
working for the class of graphs of size $n$ takes time $\Omega(n)$ on graph $H_n$.
\end{theorem}

\begin{proof}
Consider the following family of graphs $G_k$, for $k \geq 2$.
The construction of graph $G_k$ starts from a ring with $(2^{k}\cdot 5) - 4$ nodes , where all edges have port numbers 0 and 1 at the endpoints. 
All nodes in $G_k$ will have degree 3 at the end of the construction. Hence all edges that do not belong to the original ring will have port number 2 at both endpoints.
 
Pick an edge $\{u,v\}$ of the ring and call $u$ its first node. Consider the whole ring as a segment $S_k$ of $(2^k \cdot 5)-3$ edges having edge $\{u,v\}$ as its first  and last edge (hence nodes $u$ and $v$ appear twice in $S_k$).

\begin{figure}
\begin{center}
\begin{tabular}{cc}
(a)&\scalebox{0.4}{
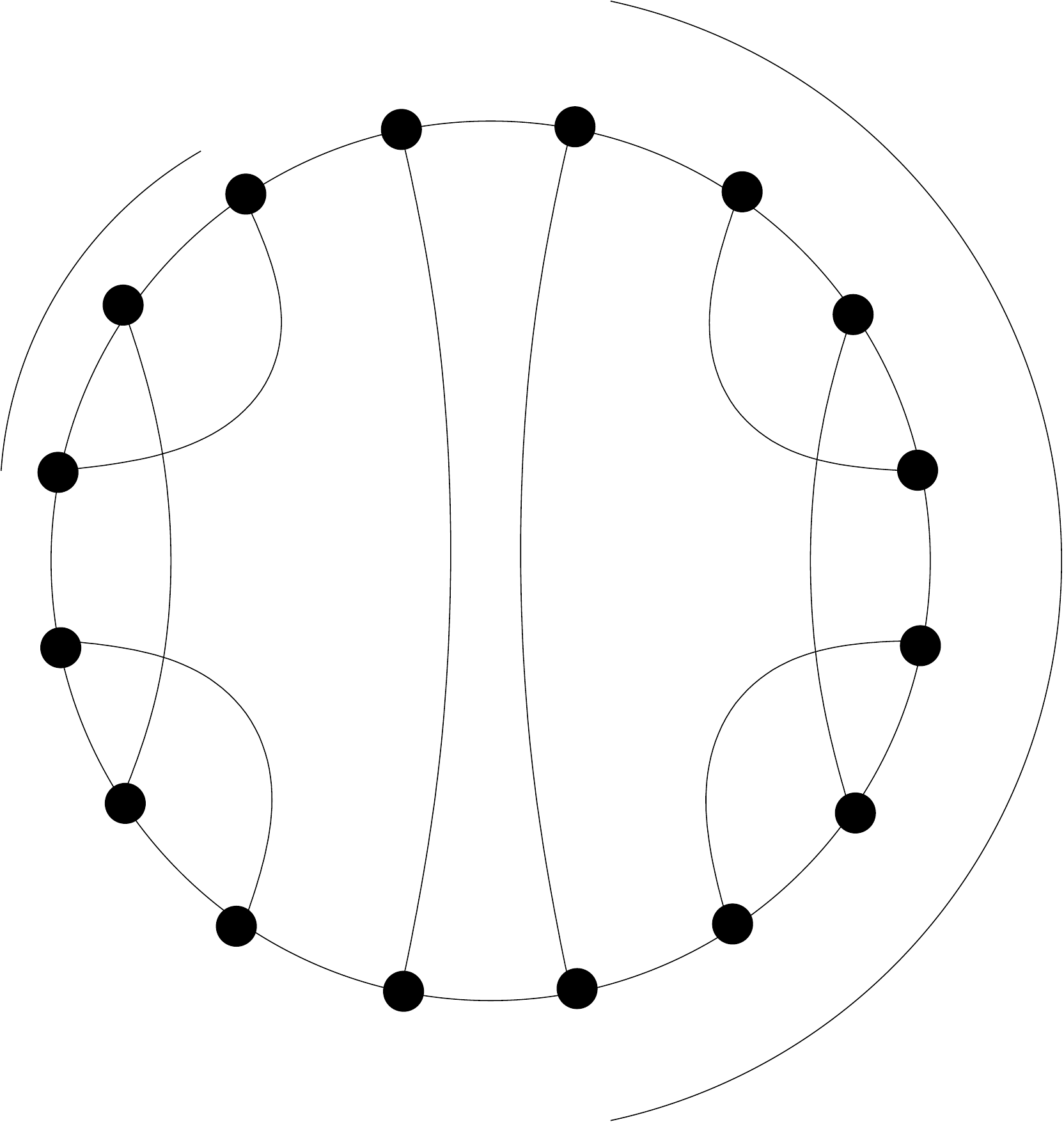}\\\\\hline\\
(b)&
\scalebox{0.4}{
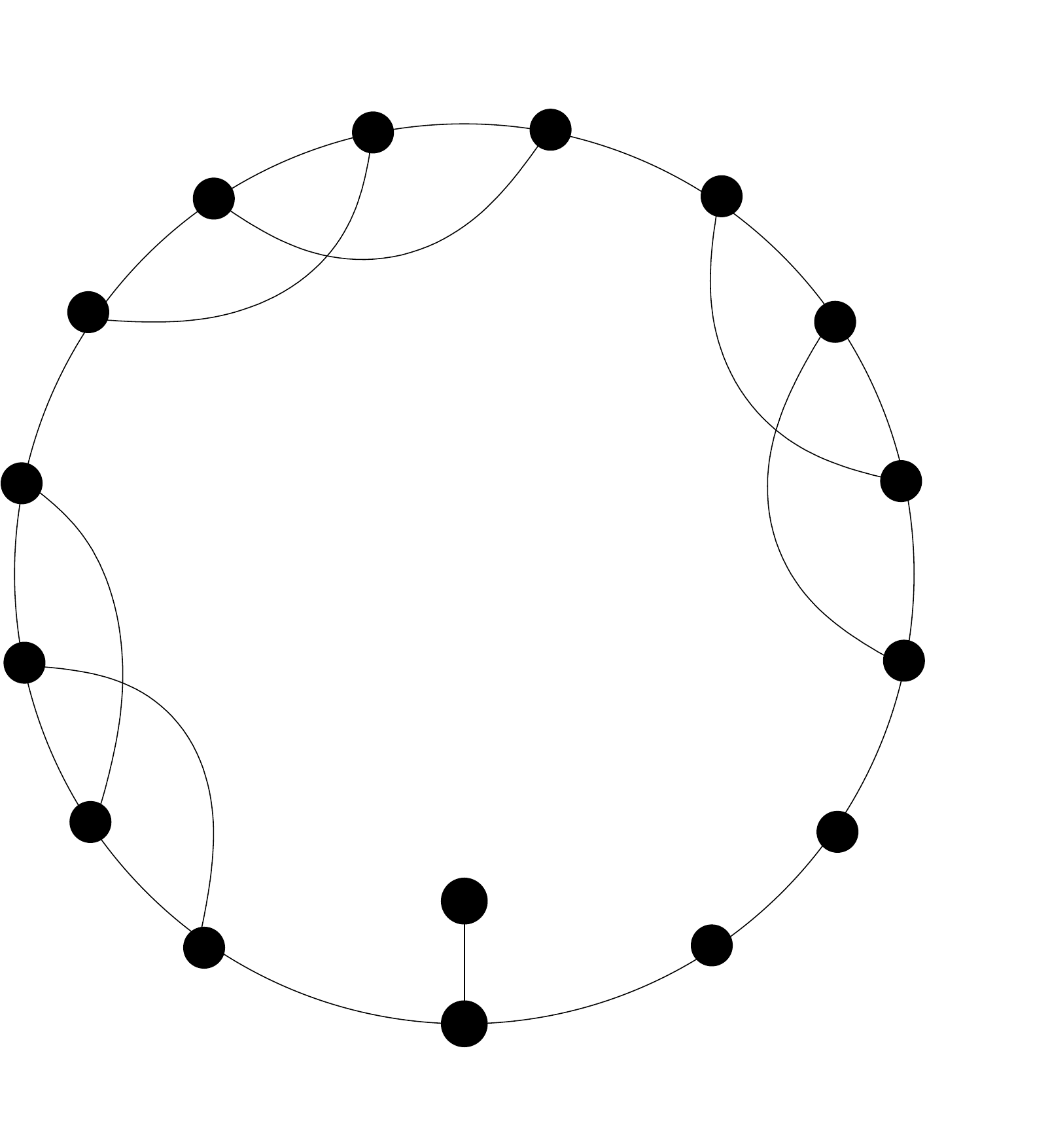}\\
\end{tabular}
\caption{\label{fig.G2} a) Graph $G_2$. b) Graph $G_{2}'$.}
\end{center}
\end{figure}


Let $\{x,y\}$ be the central edge of segment $S_k$, where $x$ precedes $y$ in $S_k$.
Connect the second endpoint of the first edge of $S_k$ (i.e., node $v$) to $x$. Connect the first endpoint of the last edge of $S_k$ (i.e., node $u$) to $y$.
Let $S_{k-1} '$ be the segment obtained from $S_k$ by removing the first edge, edge  $\{x,y\}$ and all subsequent edges.
Let $S_{k-1}''$ be the segment obtained from $S_k$ by removing all edges preceding edge $\{x,y\}$ (included) and the last edge.
If $S_{k-1}'$ and $S_{k-1}''$ have length larger than 2, proceed recursively on both segments.
If $S_{k-1}'$ and $S_{k-1}''$ have length 2, connect their central nodes.
This concludes the construction of graph $G_k$ (see Fig.~\ref{fig.G2}.a for an example). For $n=(2^{k}\cdot 5) - 4$ we define $H_n=G_k$.

Since the length of a segment $S_i$ is given by $|S_{i-1}'| + |S_{i-1}''| + 3 = 2|S_{i-1}| + 3$ and $|S_0|= 2$, we get the formula $|S_i|= (2^i \cdot 5)-3$.
Hence every segment $S_i$ has a central edge if $|S_i|>2$. 
The diameter of graph $G_k$ is at most $4k+2$. Indeed, starting from node $u$ it is possible to reach any other node in $G_k$ using at most $2k + 1$ edges. 
Views of all nodes in $G_k$ are identical, hence leader election in graph $G_k$ is impossible.

Consider the family of graphs $G_{k}'$, obtained, for any $k \geq 2$, from a ring with $(2^{k}\cdot 5) - 5$ nodes, where all edges have port numbers 0 and 1 at the endpoints, as follows.
Add a new node $v$ and connect it to some node $u$ in the ring;
divide the remaining $(2^{k}\cdot 5) - 6$ nodes of the ring into $\lfloor ((2^{k}\cdot 5) - 6)/4 \rfloor$ groups of 4 consecutive nodes and connect the first node of each group with the third one of the same group, as well as the second with the fourth.
Again, new edges joining nodes of the original ring have port number 2 at both endpoints.  
Leave the remaining 2 consecutive nodes  unchanged (see Fig.~\ref{fig.G2}.b for an example).

Suppose, for contradiction, that some algorithm $A$ correctly solves strong LE for all graphs of known size $n=(2^k \cdot 5)-4$, where $k=2,3,...$, in time at most $n/5$. In particular, for any $k\ge 5$, algorithm $A$ must report  ``LE impossible'' in graph $G_k$ in time at most $2^k-1$.
However, any node $w$ at distance at least $2^k+3$ from the unique node of degree 1 in graph $G_{k}'$ has the same view, up to depth $2^k-1$, as nodes in $G_k$. Thus node $w$ would report ``LE impossible'' in graph $G_{k}'$ and algorithm $A$ would fail to elect a leader in the solvable graph $G_{k}'$, contradicting the assumption that $A$ solves strong LE for all graphs of known size $n$, for $n=(2^k \cdot 5)-4$. This contradiction implies that every algorithm solving strong LE with known size $n$ must take time $\Omega(n)$ on graphs $H_n$ which have diameter $O(\log n)$ and level of symmetry 0.
\end{proof}

\begin{algorithm*}

\caption{SLE-known-size-and-diameter$(n,D)$\label{alg:SLEknownSizeAndDiameter}}

{\bf for} $i:=0$ {\bf to} $D-1$ {\bf do} $COM(i)$\\
compute $|\Pi_0|$; $j:=0$\\
{\bf repeat}\\
\hspace*{1cm}$COM(D+j)$; $j:=j+1$; construct $\Pi_j$\\
{\bf until} $|\Pi_j| = |\Pi_{j-1}|$\\
{\bf if} $|\Pi_j|<n$ {\bf then} report ``LE impossible''\\
{\bf else}\\
\hspace*{1cm}$V:=$ the set of nodes $v$ in $\cV^{D+j}(u)$ having the smallest $\cV^{j-1}(v)$\\
\hspace*{1cm}elect as leader the node in $V$ having the lexicographically smallest path from $u$

\end{algorithm*}

{We finally show that if both $D$ and $n$ are known, then the optimal time of strong LE is $\Theta(D+\lambda)$, for graphs with level of symmetry $\lambda$.
The upper bound is given by Algorithm~\ref{alg:SLEknownSizeAndDiameter} - SLE-known-size-and-diameter $(n,D)$.
The algorithm is a variation of Algorithm~\ref{alg:knownDiameter} -  WLE-known-diameter$(D)$ - with an added test on the size of the partition $\Pi_j$ after exiting the ``repeat'' loop.}.
%
%
%
%
%
%


The following result says that Algorithm SLE-known-size-and-diameter$(n,D)$ is fast. In fact, compared to Theorems \ref{impos} and \ref{lower-strong}, 
it shows that while knowledge of the diameter alone does not help to accomplish strong LE, when this knowledge is added to the knowledge of the size,
it may exponentially decrease the time of strong LE.

\begin{theorem}\label{all}
Algorithm  SLE-known-size-and-diameter\-$(n,D)$ performs strong LE in the class of graphs of size $n$ and diameter $D$, in time $O(D+\lambda)$,
for graphs with level of symmetry $\lambda$.
\end{theorem}

\begin{proof}
The time complexity $O(D+\lambda)$ of Algorithm SLE-known-size-and-diameter$(n,D)$ on solvable graphs follows directly from the analysis of  {Algorithm~\ref{alg:knownDiameter}}.
The test on the number of sets in partition $\Pi_j$ correctly identifies solvable and non-solvable graphs. Indeed, in view of Proposition~\ref{stop}, $\Pi_j=\Pi$ and 
it is a direct consequence of Proposition~\ref{possible} that leader election is possible on an $n$-node graph, if and only if, $\Pi$ contains $n$ sets.
Proposition~\ref{Lambda} guarantees completion time $O(D+\lambda)$ on non-solvable graphs, while uniqueness of the elected leader on solvable graphs follows from the analysis of {Algorithm~\ref{alg:knownDiameter}}.
This concludes the proof.
\end{proof}

Since the lower bound in Theorem \ref{lower-weak} was formulated for known $n$, $D$ and $\lambda$, it implies a matching lower bound for the optimal time of strong LE
with known $n$ and $D$, showing that this time is indeed $\Theta(D+\lambda)$ for graphs with level of symmetry $\lambda$. 

\section{Conclusion}

We established the optimal time of weak and strong leader election, depending on the knowledge of the size $n$ and of the diameter $D$ of the graph.
For each scenario the upper bounds were shown by proposing an algorithm, and matching lower bounds were proved. The optimal time turned out to depend
on the knowledge of one or both of the parameters $n$ and $D$. 

Notice that the comparison of assumptions for our matching upper and lower bounds in various scenarios of weak and strong leader election shows that,
while the level of symmetry may significantly influence optimal election time, the {\em knowledge} of this level is not important. Indeed, for each task
and for various combinations of knowledge of the size $n$ and/or the diameter $D$, adding the knowledge of the symmetry level $\lambda$ does not help.
More precisely, weak LE is impossible knowing neither $n$  nor $D$, even if knowledge of $\lambda$, e.g., $\lambda=0$ is added, cf. \cite{YK3}. On the other hand,
Theorems \ref{weak-known-diam}, \ref{weak-known-size}, and \ref{lower-weak} show that with known $D$ or $n$, optimal time of weak LE is 
$\Theta(D+\lambda)$, regardless of whether $\lambda$ is known or not (but the time itself depends on both $D$ and $\lambda$). The same is true for strong LE. Indeed,
the impossibility result in Theorem \ref{impos} and the lower bound from Theorem \ref{lower-strong} hold even for  known $\lambda=0$, while the upper bound 
from Theorem \ref{strong-size} does not require the knowledge of $\lambda$. For known $D$ and $n$, the upper bound $O(D+\lambda)$ holds without knowing $\lambda$
by Theorem \ref{all}, while the lower bound carried over from Theorem \ref{lower-weak} holds even when $\lambda$ is known. 

It would be interesting to investigate other complexity measures of the leader election problem, such as the bit or message complexity of communication
needed to accomplish it, in relation to the knowledge of the three parameters of the network used in our present study.

\bibliographystyle{splncs}


\end{document}